%% file: root.tex
\documentclass[letterpaper, 10 pt, conference]{ieeeconf}  % Comment this line out if you need a4paper

\IEEEoverridecommandlockouts                              % This command is only needed if 
\usepackage[caption=false]{subfig}% Support for small, `sub' figures and tables

\usepackage{amssymb}
\usepackage{bbold}
\usepackage{amsmath}
\usepackage{color}
\usepackage{graphicx}
\usepackage{epsfig}
\usepackage{color}
\usepackage{epic}
\usepackage{eepic}
\usepackage{pdfpages}

\newcommand{\Tinverse}{{T^{-1}}}
\newcommand{\AAA}{{\mathbb A}}
\newcommand{\PP}{{\mathbb P}}
\newcommand{\UU}{{\mathbb U}}

\newcommand{\eps}{{\varepsilon}}
\newcommand{\algebra}{{\mathcal B}}
\newcommand{\foralmost}{{\text{ for a.e. }}}
\newcommand{\one}{{\mathbb{1}}}
\renewcommand{\restriction}{|}

\newtheorem{theorem}{Theorem}[section]
\newtheorem{lemma}[theorem]{Lemma}

\newtheorem{definition}[theorem]{Definition}
\newtheorem{example}[theorem]{Example}

\newtheorem{remark}{Remark}

\title{\LARGE \bf
On the Almost Global Stability of Invariant Sets*
}

\author{\"Ozkan Karabacak$^{1,2}$, Rafael Wisniewski$^{2}$ and John Leth$^{2}$% <-this % stops a space
\thanks{*This work has been supported by the Independent Research Fund Denmark (DeBaTe) and by the Scientific and Technological Research Council of Turkey (111E475)}% <-this % stops a space
\thanks{$^{1}$Department of Electronics and Communication Engineering, Istanbul Technical University, 34469 Istanbul, Turkey
         {\tt\small ozkan2917@gmail.com}}%
\thanks{$^{2}$Department of Electronic Systems, Automation and Control, Aalborg University, Fredrik Bajers Vej 7 C, 9220 Aalborg East, Denmark
        {\tt\small ozk@es.aau.dk, raf@es.aau.dk, jjl@es.aau.dk}}%
}

\begin{document}

\maketitle
\thispagestyle{empty}
\pagestyle{empty}

%%%%%%%%%%%%%%%%%%%%%%%%%%%%%%%%%%%%%%%%%%%%%%%%%%%%%%%%%%%%%%%%%%%%%%%%%%%%%%%%
\begin{abstract}
For a given invariant set of a dynamical system, it is known that the existence of a Lyapunov-type density function, called Lyapunov density or Rantzer's density function, may imply the convergence of almost all solutions to the invariant set, in other words, the almost global stability (also called almost everywhere stability) of the invariant set. For discrete-time systems, related results in literature assume that the state space is compact and the invariant set has a local basin of attraction. We show that these assumptions are redundant. Using the duality between Frobenius-Perron and Koopman operators, we provide a Lyapunov density theorem for discrete-time systems without assuming the compactness of the state space or any local attraction property of the invariant set. As a corollary to this new discrete-time Lyapunov density theorem, we provide a continuous-time Lyapunov density theorem which can be used as an alternative to Rantzer's original theorem, especially where the solutions are known to exist globally.
\end{abstract}

\input{./Intro.tex}

\input{./Preliminaries.tex}
\input{./Discrete.tex}

\input{./Continuous.tex}
\input{./Examples.tex}

\input{./Conclusion.tex}

\bibliographystyle{IEEEtran}
\bibliography{measuring}
\end{document}

%% file: Intro.tex
\section{Introduction}

Stability of invariant sets of dynamical systems with respect to perturbations in initial conditions has always been a crucial topic in systems and control. A control task aims at global asymptotic stability of the invariant set under consideration, which will ensure the resilience of dynamic behaviour with respect to all changes in initial conditions. Recently, as an alternative to global asymptotic stability, a weaker notion called \emph{almost global stability} (also called \emph{almost everywhere stability}), which describes convergence of almost all solutions to the invariant set, has been introduced and found some applications in nonlinear analysis \cite{Monzon2005,Monzonthesis,Vaidya2008,Potrie2009,Zuyev2015,Rajaram2015} and control engineering \cite{Rantzer2001proc,Prajna2004,Vaidya2010}. 
The main idea was first introduced by Milnor \cite{Milnor1985}, where a new notion of attractor is proposed as a minimal\footnote{It is minimal in the sense that no other proper subset of it attracts almost the same set of initial states.} closed invariant set that attracts a set of positive measure. Milnor proved that if the state space is compact, there always exists a maximal attractor, called the \emph{likely limit set}, which attracts almost all initial conditions. Obviously, if the state space is not compact, such a maximal attractor may not exist. 
A natural question along this line is to find sufficient conditions, even for systems with a non-compact state space, that ensure convergence of almost all solutions to a given invariant set, hence the almost global stability of the invariant set. Rantzer \cite{Rantzer2001} introduced a dual Lyapunov approach to the almost global stability problem, where the existence of a density function, satisfying some Lyapunov-like conditions, is shown to imply almost global stability of an equilibria. This density function, called Lyapunov density or Rantzer's density function, has been effectively applied to nonlinear feedback control \cite{Prajna2004} of continuous-time systems. For discrete-time systems, Vaidya and Mehta \cite{Vaidya2008} {use the Frobenius-Perron operator to} provide a counterpart of Rantzer's Lyapunov density theorem for invariant sets, assuming that the state space is compact and the invariant set is almost locally stable, namely that it attracts almost every point in some neighborhood of it. {Using the same assumptions, a Lyapunov density theorem for invariant sets has been given for the continuous-time case in \cite{Rajaram2010}.}
%\footnote{In \cite{Vaidya2008}, this assumption of local almost everywhere stability has been wrongly referred/creditted \ozk{a better verb exists} to Milnor} 
Also, the discrete-time Lyapunov density theorem in \cite{Vaidya2008} has been applied to a feedback control design \cite{Vaidya2010}.
Using the Koopman operator, another approach for the global stability of nonlinear systems is given in \cite{Mauroy2013,Mauroy2014}, which provides sufficient conditions for the global asymptotic stability of an equilibrium. 
Another related concept called occupation measure has been utilised in \cite{henrion2014} for the convex computation of the region of attraction.

In this paper, {we use the duality between Frobenius-Perron and Koopman operators (see \cite{Foguel1969} and \cite{Lasota1994} for the general theory of Markov processes where these operators appear) and} provide a discrete-time version of Rantzer's continuous-time density theorem \cite{Rantzer2001} for invariant sets, that do not assume the compactness of the state space and the local stability of the invariant set \footnote{It is claimed in \cite{Rantzer2001} that Theorem~2 in \cite{Rantzer2001} can be viewed as a discrete-time counterpart for the Lyapunov density theorem (Theorem~1 in \cite{Rantzer2001}). However, this is not straightforward, and the complete discrete-time Lyapunov density theorem is given  in \cite{Vaidya2008} with much stronger assumptions.}. Subsequently, we use our theorem for discrete-time to prove a generalization of the continuous-time density theorem \cite{Rantzer2001} to invariant sets \footnote{For continuous-time, such a generalization appeared in \cite{Rantzer2001proc} under a boundedness assumption on the vector field. See also \cite{Zuyev2015} for another generalization.}. Both theorems require very similar sets of conditions for a Lyapunov density function, leading to a unification of the theories for discrete- and continuous-time systems, which we summarize below postponing the detailed descriptions for the more general statements of the theorems to Sections~\ref{Sec:Discrete} and \ref{sec:continuous}.  

Let us consider a dynamical system on $\mathbb R^n$ given by
\begin{equation}
\label{eq:general}
x^+(t)=F(x(t)), \ t\in\mathbb T,
\end{equation}
where the time set $\mathbb T$ is either $\mathbb R$ or $\mathbb N$, and $x^+(t):=\frac{dx(t)}{dt}$ if $\mathbb T=\mathbb R$ and  $x^+(t):=x(t+1)$ if $\mathbb T=\mathbb N$. For $\mathbb T=\mathbb N$, we 
%write $T$ in place of $\gamma$ and 
assume that $F$ is nonsingular. For $\mathbb T = \mathbb R$, 
% namely it maps positive measure sets to positive measure sets (see Subsection~X for a more formal definition), 
we 
%write $F$ in place of $\gamma$ and 
assume that $F$ is locally Lipschitz and that solutions of \eqref{eq:general} uniquely exist for all time and for all initial conditions. 
{As in \cite{Vaidya2008} and \cite{Vaidya2010}, we use Frobenius-Perron operator $\PP$ to characterize Lyapunov density in discrete-time. $\PP$ captures the evolution of a distribution of initial states under the discrete-time dynamics. Similarly, the evolution of distributions in continuous-time is captured by the infinitesimal operator $\AAA$. Precise definitions for these operators will be given in the sequel.
An invariant set $A\subset \mathbb R^n$ is said to be \emph{almost globally stable} if there exists a subset $N\subset \mathbb R^n$ with zero Lebesgue measure such that $x(t)\to A$ as $t\to\infty$ for all $x(0)\in N^c$.
}

For an invariant set $A\subset \mathbb R^n$, we say that a continuously differentiable real function $\rho$ 
%(for discrete-time, $\rho$ can be only measurable) 
defined on $A^c$ is a Lyapunov density for $A$, if it is positive almost everywhere on $A^c$, integrable away from $A$ and properly subinvariant on $A^c$. Specifically, we say that $\rho$ is properly subinvariant  on $A^c$ if $\mathbb L\rho(x)< 0$ for almost all $x \in A^c$, where $\mathbb L\rho = \PP\rho - \rho$ for discrete time systems and $\mathbb L\rho = \AAA\rho$ for continuous time systems.

%namely $\mathbb L\rho(x)<\rho(x)$ for almost all $x\in A^c$, where $\mathbb L$ is the Frobenious-Perron operator for discrete-time systems and the infinitesimal operator of the semigroup of Frobenious-Perron operators for continuous-time systems. 
Our main result can now be stated in a unified form for discrete- and continuous-time systems as follows: \newline
\textbf{Main result:} A compact invariant set $A$ is almost globally stable if there exists a Lyapunov density for $A$.

Section~\ref{sec:preliminaries} summarizes the required preliminary definitions and results on Frobenius-Perron operator, Koopman operator and their duality. The detailed descriptions and proofs of the main result for discrete time and continuous timme are given in Section~\ref{Sec:Discrete} and Section~\ref{sec:continuous}, respectively. Finally, some illustrative examples of applications of the main result are provided in Section~\ref{sec:examples}.

%% file: Preliminaries.tex
%By a measure $\mu$ on a metric space $\mathbb R^n$, we always mean a $\sigma$-finite measure defined on the Borel $\sigma$-algebra of $\mathbb R^n$. The Lebesgue measure will be denoted by $m$. We say that a property is satisfied $\mu$ almost everywhere of $\mu$-a.e. in short, if the set of points (initial conditions for solutions) that do not satisfy the property has $\mu$-measure zero. For any set $V\subset \mathbb R^n$, we denote the $\eps$-neighborhood of $V$ by $V_\eps$. $V^c$ denotes the complement of $V$.

\section{Preliminary Definitions and Tools}
\label{sec:preliminaries}
By a measure $\mu$ on $\mathbb R^n$, we mean a $\sigma$-finite measure defined on the Borel $\sigma$-algebra $\mathcal B$ of $\mathbb R^n$. In particular, Lebesgue measure on $\mathbb R^n$ will be denoted by $m$. A measure $\mu$ is said to be absolutely continuous (with respect to $m$) if $\mu(A)=0$ whenever $m(A)=0$. 

We assume that the dynamics given by $T:\mathbb R^n\to \mathbb R^n$ is nonsingular; namely, $A\in\algebra$ and $m(A)=0$ implies that $m(\Tinverse A)=0$.

%A pair of measure $\mu_1$ and $\mu_2$ are said to be equivalent if they give rise to the same set of zero measures, i.e. if for any $A\in\algebra$, $\mu_1(A)=0 \iff \mu_2(A)=0$. 
%A measure $\mu$ is said to be properly subinvariant if $\mu(\Tinverse A)<\mu(A)$ whenever $\mu(A)>0$.

%A measure $\mu_2$ is said to be weaker than another measure $\mu_1$ if $\mu_1(A)=0 \implies \mu_2(A)=0$. In this case, we write $\mu_2\ll \mu_1$. Clearly, $\mu_1$ and $\mu_2$ are equivalent if and only if $\mu_1\ll\mu_2$ and $\mu_2\ll\mu_1$.
%Radon-Nikodym theorem states that if $\mu$ and $m$ are signed measures with $\mu\ll m$ then there exists a measurable function defined on $\mathbb R^n$ such that
The Radon-Nikodym theorem states that if $\mu$ is absolutely continuous then there exist a unique (modulo sets of zero Lebesgue measure) nonnegative measurable function $\rho$ such that   
\begin{equation}
 \label{eq:RadonNikodym}
 \mu(A)=\int_A \rho\ dm.
\end{equation}
$\rho$ is called the Radon-Nikodym derivative of $\mu$ (with respect to $m$).
%If $\mu$ is a positive measure equivalent to a positive measure $m$ then $\rho$ must be strictly positive ($m$.a.e.). 
%The Radon-Nikodym derivative of a measure is unique up to a set of $m$-measure zero. 
%On the other hand, if $\rho\geq 0$ ({\em resp}. $\rho>0$) is an $m$-measurable function on $\mathbb R^n$, then $\mu(A):=\int_A\rho\ dm$ defines a measure weaker than ({\em resp}. equivalent to) $m$. 
On the other hand, if $\rho$ is a nonnegative measurable function then $\mu(A):=\int_A\rho\ dm$ defines an absolutely continuous measure on $\mathbb R^n$. Therefore, there is a one-to-one correspondence between the set of absolutely continuous measures on $\mathbb R^n$ and the set of equivalence classes of non-negative measurable functions on $\mathbb R^n$, where equivalence classes are defined as sets of functions that differs only on sets of zero Lebesgue measure. Both of these sets are denoted by $\mathcal M^+(\mathbb R^n)$. Naturally, this extends to another equivalence between absolutely continuous signed measures on $\mathbb R^n$ and the set of equivalence classes of measurable functions on $\mathbb R^n$, both of which are denoted by $\mathcal M(\mathbb R^n)$. In addition, the linear vector space of all signed measures on $\mathbb R^n$ is denoted by $\overline{\mathcal M}(\mathbb R^n)$ and the space of equivalence classes of integrable functions on $\mathbb R^n$ is denoted by $\mathcal L_1(\mathbb R^n)$. Note that $\overline{\mathcal M}(\mathbb R^n)\supset \mathcal M(\mathbb R^n)\supset 
%\mathcal M^+(\mathbb R^n)\supset 
\mathcal L_1(\mathbb R^n)$.
%Similarly, the set of finite measures  are in one-to-one correspondence with po $\mathcal L_1$, i.e. the set of (equivalence classes of) integrable real valued functions on $\mathbb R^n$ 

\subsection{Frobenius-Perron Operator}
\label{sec:Perron}
%Let $\mathcal M$ denote the linear vector space of signed measures on $\mathbb R^n$.
The evolution of distributions under the discrete-time dynamics of \eqref{eq:general} can be captured by a linear operator $\PP:\overline{\mathcal M}(\mathbb R^n)\to \overline{\mathcal M}(\mathbb R^n)$ (called Frobenius-Perron operator) defined as
\begin{equation}
\label{eq:Perrononmeasures}
 (\PP \mu)(A):=\mu(\Tinverse A).
\end{equation}
Assume that $\mu$ is absolutely continuous. Since $T$ is nonsingular 
\begin{multline}
m(A)=0 \implies m(\Tinverse A)=0\\ \implies \mu(\Tinverse A)=0\implies (\PP\mu)(A)=0
\end{multline}
so $\PP\mu$ is also absolutely continuous.
Therefore, $\PP$ maps $\mathcal M^+(\mathbb R^n)$ to itself and similarly $\PP$ maps $\mathcal M(\mathbb R^n)$ to itself.
%Equivalently, $\PP$ can be seen to act over the space of measurable functions (Radon-Nikodym derivatives of measures). 
For a nonnegative measurable function $\rho\in \mathcal M^+(\mathbb R^n)$, $\PP\rho$ is defined as the Radon-Nikodym derivative of $\PP\mu$ with respect to $m$. Using \eqref{eq:RadonNikodym} and \eqref{eq:Perrononmeasures}, it follows that\footnote{We allow here the integral to be infinite.}
\begin{equation}
\label{eq:Perronondensities}
 \int_A \PP \rho dm =\int_{\Tinverse A}\rho dm.
\end{equation}
%This action can be found as 
%\begin{equation}
%\label{eq:Perronondensities}
% \PP \rho =\rho(\Tinverse x)/J(x),
%\end{equation}
%where the function $J$ is the density of $m(\Tinverse A)$ with respect to $m$ (see \cite{Lasota1994}).           
%If $\rho$ is finite then necessarily $\PP\mu$ is finite. 
Substituting $A=\mathbb R^n$ in \eqref{eq:Perronondensities} implies that $\PP$ maps integrable functions to integrable functions. Hence, the restriction of $\PP$ to $\mathcal L_1(\mathbb R^n)$, denoted by $\PP\restriction_{\mathcal L_1(\mathbb R^n)}$,  produces a map $\mathcal L_1(\mathbb R^n)\to\mathcal L_1(\mathbb R^n)$ which is a Markov operator. Namely, $\PP=\PP\restriction_{\mathcal L_1(\mathbb R^n)}$ satisfies the following conditions:

$\rho\geq0 \implies \PP\rho\geq0$ and $\|\PP\rho\|=\|\rho\|.$
%Hence, the theory of Markov processes can be applied to $\PP$ (see Subsection~\ref{sec:Markov}).
\subsection{Koopman Operator}
\label{sec:Koopman}
A dual method to capture the statistical behaviour of the deterministic system \eqref{eq:general} is via the Koopman operator $\UU$, which describes the evolution of the values of observables under the dynamics of \eqref{eq:general}.
% Let $\mathcal O=\mathcal O(\mathbb R^n)$ denote the set of
%\footnote{Here, we call $f$ observable to emphasize that we shall think $f$ as attached to the space rather than to the state points that are evolving. If one assumes the contrary, namely assumes that the values of $f$ is transfered by state transitions, then $\UU$ seems to propogate $f$ backward in time.}) 
%equivalence classes of measurable functions on $\mathbb R^n$. 
Define $\UU:\mathcal M(\mathbb R^n)\to \mathcal M(\mathbb R^n)$ as
\begin{equation}
\label{eq:Koopmanonmeasurables}
 (\UU f)(x):=f(T x).
\end{equation}
Clearly, $\UU$ is linear and maps positive functions to positive functions. It also maps bounded functions to bounded functions. Hence $\UU$ can be restricted to $\mathcal L_\infty$, the normed vector space of equivalence classes of essentially bounded measurable functions. 

\subsection{Duality between Frobenius-Perron and Koopman Operators}
\label{sec:duality}
Let the scalar product\footnote{That is, a non-degenerate bilinear function in $\mathcal L_1\times\mathcal L_\infty$, see \cite{greub}.} $\langle\cdot,\cdot\rangle$ between $\mathcal L_1$ and $\mathcal L_\infty$ be defined by 
\begin{equation}
 \label{eq:bracket}
 \langle\rho,f\rangle:=\int\rho f\ dm, \quad\rho\in \mathcal L_1,~~f\in\mathcal L_\infty.
\end{equation}
Note that since $\rho\in \mathcal L_1,~~f\in\mathcal L_\infty$, we have $\langle\rho,f\rangle\leq\text{ess}\sup f\int\rho dm<\infty$. The duality between $\PP:\mathcal L_1\to\mathcal L_1$ and $\UU:\mathcal L_\infty\to \mathcal L_\infty$ with respect to $\langle\cdot,\cdot\rangle$ can be seen by first observing that
\begin{multline}\label{eq:dualityEQ}
\langle\PP\rho,{1}_A\rangle=\int \PP \rho 1_A \ dm = \int_A \PP \rho \ dm
=\int_{\Tinverse A}\rho\ dm\\=\int \rho 1_{\Tinverse A}\ dm=\langle\rho,1_{\Tinverse A}\rangle=\langle\rho,\UU 1_A\rangle, 
\end{multline}
where $1_A$ denotes the characteristic function for $A\subset \mathbb R^n$.
Since any function in $\mathcal L_\infty$ can be approximated by characteristic functions we have the following duality
\begin{equation}
 \label{eq:duality}
 \langle\PP \rho, f\rangle=\langle\rho,\UU f\rangle, \quad \rho\in\mathcal L_1, \ f\in\mathcal L_\infty.
\end{equation}
%This duality persists even for general measurable functions \cite[Theorem~5.2]{Cinlar2001}.
\begin{remark}\label{re:duality}
By \eqref{eq:bracket} and \eqref{eq:dualityEQ}, this duality persists even for general measurable functions when considering $\PP$ and $\UU$ as functions ${\mathcal M}(\mathbb R^n)\to {\mathcal M}(\mathbb R^n)$.
\end{remark}

\subsection{An invariant set $A$ and $\PP$ restricted to $A^c$}

Let $A\subset \mathbb R^n$ be an invariant set of the dynamics $F$, that is $F(A)\subset A$. We observe that for any set $V\subset A^c$, $F^{-1}(V)\subset A^c$, $\PP$ maps $\overline{\mathcal M}(A^c)$ to itself. Moreover, since $A^c$ may not be invariant, the (restricted) operator  $\PP\restriction_{A^c}$ may not be a Markov operator, but it is a sub-Markov operator, namely it satisfies the following conditions:

$\rho\geq0 \implies \PP\restriction_{A^c}\rho\geq0$ and $\|\PP\restriction_{A^c}\rho\|\leq\|\rho\|$.

Let $A_\eps$ denote the $\eps$-neighborhood of $A$, namely $A_\eps=\{x\in \mathbb R^n:\inf_{y\in A} d(x,y)<\eps\}$, where $d$ denotes the standard metric on $\mathbb R^n$. We also write $A_\eps^c$ for $(A_\eps)^c$.
We say that $\rho\in\mathcal M^+(A^c)$ is integrable away from $A$ if $\int_{A_\eps}\rho(x)\;m(dx)$ is finite for all $\eps>0$. $\rho\in\mathcal M^+(A^c)$ is said to be properly subinvariant on $A^c$ if $\PP\restriction_{A^c}\rho(x)<\rho(x)$ for almost all $x\in A^c$.
 
%\subsection{Some Facts From the Ergodic Theory of Markov Processes}
% \label{sec:Markov}

%In this section, all equalities, inequalities and inclusions are valid up to a set of measure zero.

%The following is a consequence of Hopf maximal ergodic theorem:  

%\begin{definition}[Hopf decomposition, \cite{Foguel1969},page 11]
%\label{def:Hopf}
% Let $\PP$ be a Markov operator on $\mathbb R^n$ and $0<\rho\in \mathcal L_1(\mathbb R^n)$. The conservative set $C\subset \mathbb R^n$ is defined as 
% $$C:=\{x\in \mathbb R^n \mid \sum_{k=0}^\infty \PP^k\rho(x)=\infty\}.$$
% The dissipative set is defined as $D:=\mathbb R^n-C$.
%\end{definition}

%\begin{lemma}[\cite{Foguel1969},page 11,15 and 17]
%\label{lem:Foguel}
%\begin{itemize}
% \item[a)] The definitions of dissipative and conservative sets are independent from the choice of the function $\rho$. 
% \item[b)] For any $0\leq \rho \in \mathcal L_1(\mathbb R^n)$, $\sum_{k=0}^\infty \PP^k\rho(x)<\infty$ if $x\in D$. 
% \item[c)] If $\rho>0$ and $\PP\rho\leq \rho$ then $\PP\rho=\rho$ on $C$.
% \item[d)] If $\mathbb R^n=D$, then $\bar\rho:=\sum_{k=0}^\infty \PP^k\rho(x)$ defines a $\sigma$-finite measure $\mu(A):=\int_{A} \bar\rho(x) dm$ that is properly subinvariant, namely satisfies $(\PP\mu)(A)<\mu(A)$ for any $A$ with $0<\mu(A)<\infty$.
%\end{itemize}
%\label{lem:Foguel}
%\end{lemma}

%% file: Discrete.tex
\section{Discrete-Time Case}
\label{Sec:Discrete} 
In this section, we consider almost global stability of an invariant set $A$ for a discrete-time system. In particular, we prove that $A$ is almost globally stable if there exists a density $\rho$ that is positive, integrable away from $A$ and properly subinvariant on $A^c$, that is, a discrete-time Lyapunov density as defined below. For the definitions of properly subinvariant density and integrability away from a subset, see the previous subsection.

We consider the dynamical system \eqref{eq:general} with $\mathbb T=\mathbb N$ on a $\sigma$-finite measure space $(X,\algebra,\mu)$ 
%given by
%\begin{equation}
 %%\label{sys:discrete}
 %x(k+1)=T(x(k))\quad x(0)\in X,
%\end{equation}
where $F$ is nonsingular. 
Throughout this section, we assume that $A\subset X$ is an invariant set for the system \eqref{eq:general}, $\PP$ and $\UU$ are the Frobenius-Perron and Koopman operators of the system \eqref{eq:general} restricted to $A^c$, respectively.

\begin{definition}[Discrete-time Lyapunov density]
\label{def:DiscreteLyapunovDensity}
A measurable function $\rho$ defined on $A^c$ is called a \emph{discrete-time Lyapunov density for $A$} if it satisfies the following conditions:
\begin{itemize}
	\item[]DLD1: $\rho$ is positive: $\rho(x)>0\ \foralmost x\in A^c.$
	\item[]DLD2: $\rho$ is integrable away from $A$: \begin{align}\int_{A_\eps^c}\rho(x)m(dx)<\infty\ \forall\eps>0.\end{align}
	\item[]DLD3: $\rho$ is properly subinvariant on $A^c$: \begin{align}\PP\rho(x)<\rho(x)\ \foralmost x\in A^c.\end{align}
	\end{itemize}
\end{definition}

This definition is less restrictive than the definition of a Lyapunov measure in \cite{Vaidya2008}, as it does not assume a local attraction property of the attractor. 
For discrete time systems the main result is the following: 

\begin{theorem}[Almost global stability in discrete time]
\label{th:discrete}
An invariant set $A$ of the system \eqref{eq:general} for $\mathbb T=\mathbb N$ is almost globally stable if there exists a discrete-time Lyapunov density for $A$.
\end{theorem}

{
\begin{remark}
In the language of attractor theory (see for instance \cite{Karabacak2011} and the references therein), Theorem~\ref{th:discrete} provides conditions for the existence of a maximal Milnor attractor for systems on noncompact spaces. Note that a maximal Milnor attractor, also called the likely limit set, always exists if the state space is compact \cite{Milnor1985}.
\end{remark}

\begin{remark}
 For the case where $A^c$ is also invariant, $\PP\restriction_{A^c}$ is a Markov operator, namely it preserves the total mass of a measurable function $\rho$. This together with the proper subinvariance of $\rho$ implies that the Lyapunov density $\rho$ has infinite mass, namely it is not integrable. The nonintegrability of $\rho$ makes it difficult to approximate numerically with the methods in literature, such as the moment method for approximating finite measures \cite{LassarreMomentBook2010}.  
\end{remark}
}

We will use the following characterization for almost global stability of an invariant set $A$.
%\begin{lemma}
% \label{lem:AEAKoopman}
% $T^nx\to A$ if and only if $\sum_{k=0}^\infty\UU^k1_{A_\eps^c}(x)<\infty$
%for all $\eps>0$.
%\end{lemma}

%\begin{proof}
%$\sum_{k=0}^\infty\UU^k1_{A_\eps^c}(x)=\sum_{k=0}^\infty 1_{A_\eps^c}(T^k x)$ is equal to the number of visits of the trajectory $T^nx$ to the closed set $A_\eps^c$. Hence, it is finite if and only if there exists an $N(\epsilon,x)>0$ such that $T^n(x)\in A_\eps\ \forall n>N$, namely $T^nx\to A$. 

%Therefore, $\sum_{k=0}^\infty\UU^k1_{B_\eps^c}(x_0)$ is finite $\foralmost$ if and only if the set of limit points of $x(n)$ is nonempty and contained in $B_\eps$ for $m$-a.e. initial points $x_0$. Invoking this statement for a sequence $\eps_i\to0$ and considering the fact that $\cap_i B_{\eps_i}=\{0\}$ and that any intersection of countably many full measure sets is a full measure set give the result.  
%\end{proof}

%Invoking this lemma $\foralmost x$, we obtain the following:

\begin{lemma}
\label{lem:Koopmancharacterization}
An invariant set $A$ is almost globally stable if and only if, for all $\eps>0$, $\sum_{k=0}^\infty\UU^k\one_{A_\eps^c}(x)<\infty\ \foralmost x$. 
\end{lemma}
\begin{proof}
Note that $T^nx\to A$ if and only if, for all $\eps>0$,  $\sum_{k=0}^\infty\UU^k1_{A_\eps^c}(x)<\infty$. In fact, $\sum_{k=0}^\infty\UU^k1_{A_\eps^c}(x)=\sum_{k=0}^\infty 1_{A_\eps^c}(T^k x)$ is equal to the number of visits of the trajectory $T^nx$ to the closed set $A_\eps^c$. Hence, it is finite if and only if there exists an $N(\epsilon,x)>0$ such that $T^n(x)\in A_\eps\ \forall n>N$.  namely $T^nx\to A$.

By the above statement, the necessity is trivial. To prove the sufficiency, choose a sequence of positive numbers $\{\eps_n\}\to 0$. By assumption, for each $\eps_n$, $\sum_{k=0}^\infty\UU^k\one_{A_{\eps_n}^c}(x)<\infty\ \foralmost x$, i.e. there exists a zero measure set $N_{\eps_n}$ such that $\sum_{k=0}^\infty\UU^k\one_{A_{\eps_n}^c}(x)<\infty\ \forall x\in N_{\eps_n}^c$. Let us define $N=\bigcup_n N_{\eps_n}$. Obviously, $N$ has zero measure and we will show that, for all $\eps>0$, $\sum_{k=0}^\infty\UU^k\one_{A_\eps^c}(x)<\infty\ \forall x\in N^c$. For a given $\eps>0$, we can choose an $n$ such that $\eps_n<\eps$. Then, $\sum_{k=0}^\infty\UU^k\one_{A_\eps^c}(x)<\sum_{k=0}^\infty\UU^k\one_{A_{\eps_n}^c}(x)<\infty\ \forall x\in N^c$, since $N^c\subset N_{\eps_n}^c$.  
\end{proof}

We will also use the following basic fact to prove Theorem~\ref{th:discrete}:

\begin{lemma}
\label{lem:limitdensity}
Condition DLD3 implies that $\foralmost x\in A^c$ $\lim_{n\to\infty}\PP^n\rho(x)$ exists and is smaller than $\rho(x)$.
\end{lemma}

\begin{proof}
Here, all arguments are valid $\foralmost x\in A^c$. First note that $\PP^n(\rho-\PP\rho)\geq 0$, since $\PP\rho<\rho$ and $\PP^n$ is a positive operator for all $n$. Therefore, $0\leq\PP^{n+1}\rho\leq\PP^n\rho$ which implies that $\PP^n\rho(x)$ is a decreasing sequence bounded from above by $\rho(x)$ and from below by $0$. Hence, the result follows.  
\end{proof}

\begin{proof}[Proof of Theorem~\ref{th:discrete}]
 By assumptions, there exists a measurable function $\rho>0$ such that $0<\PP\rho<\rho$ and $\rho$ is integrable on $A_\eps^c$ for every $\eps>0$. 
Define $\rho_0:=\rho-\PP\rho$. Clearly $\rho_0$ is positive $m$-a.e. and it follows that
 \begin{eqnarray*}
 \bar\rho_0&:=&\lim_{n\to\infty}\sum_{k=0}^n \PP^k \rho_0=\lim_{n\to\infty}\sum_{k=0}^n \PP^k \left(\rho-\PP\rho\right)\\
						&=&\lim_{n\to\infty} \left(\sum_{k=0}^n \PP^k \rho - \sum_{k=1}^{n+1}\PP^k\rho\right)=\lim_{n\to\infty} \left(\rho - \PP^{n+1}\rho\right)\\
            &=&\rho-\lim_{n\to\infty}\PP^{n+1}\rho.\\
 \end{eqnarray*}
By Lemma~\ref{lem:limitdensity}, the above pointwise limits exist and $\lim\PP^n\rho<\rho$ almost everywhere on $A^c$. Since $\rho$ is integrable on $A_\eps^c$, $\lim\PP^n \rho$ is also integrable on $A_\eps^c$ by Lebesgue's dominated convergence theorem. These imply that $\bar \rho_0$ is integrable on $A_\eps^c$. Therefore, using Tonelli's theorem with the counting measure on natural numbers and the duality between $\PP$ and $\UU$ (see Remark~\ref{re:duality}), we have
 \begin{eqnarray*}
\infty&>&  \langle\bar \rho_0,1_{A_\eps^c}\rangle=\sum_{k=0}^\infty\langle\PP^k\rho_0,1_{A_\eps^c}\rangle\\
  &=&\sum_{k=0}^\infty\langle\rho_0,\UU^k 1_{A_\eps^c}\rangle=\langle\rho_0,\sum_{k=0}^\infty\UU^k1_{A_\eps^c}\rangle\\
 \end{eqnarray*}
Since $\rho_0$ is positive $m$-a.e., $\sum_{k=0}^\infty\UU^k1_{A_\eps^c}$ is finite $m$-a.e. and from Lemma~\ref{lem:Koopmancharacterization}, $A$ is almost globally stable. Note that, although the last series may not be integrable it is measurable and the duality still works as discussed at the end of Subsection~\ref{sec:duality}.
\end{proof}

%% file: Continuous.tex
\section{Continuous-Time Case}
\label{sec:continuous}

In this section, we consider almost global stability of an invariant set $A$ of the system \eqref{eq:general} for $\mathbb T=\mathbb R$. Similarly to the discrete-time case, we prove that $A$ is almost globally stable if there exists a density $\rho$ defined on $A^c$ which is positive, properly subinvariant and integrable away from $A$. 

We consider the system \eqref{eq:general} with $\mathbb T=\mathbb R$ on a $\sigma$-finite measure space $(X,\algebra,\mu)$ 
%given by
%\begin{equation}
 %\label{sys:continuous}
 %\dot x(t)=F(x(t))\quad x(0)\in \mathbb R^n,\quad t\in[0,\infty)
%\end{equation}
where $F$ is continuously differentiable and solutions to \eqref{eq:general} exists for all initial conditions $x(0)$ and for all $t\geq 0$. We also assume that $A\subset X$ is a compact invariant set for the system \eqref{eq:general}, $\{\PP_t\}$ is the semigroup of Frobenius-Perron operators of the system \eqref{eq:general} restricted to $A^c$ and $\AAA$ is the infinitesimal operator of the continuous semigroup $\{\PP_t\}$, namely $\AAA\rho=-\nabla(F\rho)$ \cite{Lasota1994}. Here $\rho\in C^1(A^c,\mathbb{R}^n)$, the set of continuously differentiable functions from $A^c$ to $\mathbb{R}^n$.

Since $F$ is continuously differentiable, it is locally Lipschitz, and therefore finite-time solutions of \eqref{eq:general} vary continuously with respect to the initial states.  

\begin{definition}[Continuous-time Lyapunov density]
\label{def:ContinuousLyapunovDensity}
A $\rho\in C^1(A^c,\mathbb{R}^n)$ is called a \emph{continuous-time Lyapunov density for $A$} if it satisfies the following conditions:
\begin{itemize}
	\item[]CLD1: $\rho$ is positive: $\rho(x)>0\ \foralmost x\in A^c.$
	\item[]CLD2: $\rho$ is integrable away from $A$:
	\begin{align}\int_{A_\eps^c}\rho(x)\mu(dx)<\infty\ \forall\eps>0.\end{align}
	\item[]CLD3: $\rho$ is properly subinvariant on $A^c$:
	\begin{align}\AAA\rho(x)<0 \foralmost x\in A^c.\end{align}
	\end{itemize}
\end{definition}
Recall that, if $\rho$ is continuously differentiable, $\AAA\rho=-\nabla(F\rho)$.
%Recall that $A_\eps$ denotes the $\eps$-neighborhood of $A$, namely $A_\eps=\{x\in X:d(x,A)<\eps\}$.
\begin{theorem}[Almost global stability in continuous time]
\label{th:continuous}
An invariant compact set $A$ of the system \eqref{eq:general} for $\mathbb T=\mathbb R$ is almost globally stable if there exists a continuous-time Lyapunov density for $A$.
\end{theorem}

\begin{remark}
	The condition of being compact for $A$ can be substituted by being closed and that there exists a Lipschitz constant of $F$ on any arbitrarily small neighbourhood of $A$.	
\end{remark}

{
\begin{remark}
 Theorem~\ref{th:continuous} generalizes Rantzer's theorem for almost global stability of equilibria \cite{Rantzer2001} to invariant sets with a slight change in the integrability assumption. In Rantzer's theorem, $\frac{\rho \|F\|}{|x|}$ is assumed to be integrable away from $A$; whereas in Theorem~\ref{th:continuous}, integrability of $\rho$ (away from $A$) is required along with the assumption that the solutions exist globally. 
%In design, the latter condition could be easier to verify than the former. Moreover, as shown in 
Example~\ref{ex1} below shows that for some systems a function $\rho$ may satisfy the latter condition but not the former. We note here that the version of Rantzer's theorem  (implied by its  proof in \cite{Rantzer2001}) that assumes the integrability of $\rho$ and the boundedness of $F/|x|$ is more conservative than Theorem~\ref{th:continuous}, since the latter condition is only a sufficient condition for the global existence of solutions.
\end{remark}
}

We restate the following lemma from \cite{Rantzer2001}:
\begin{lemma}
\label{lem:Rantzer}
Consider an open set $D\subset\mathbb R^n$. Let $F,\rho\in\mathbb C^1(D,\mathbb R^n)$, where $\rho$ is integrable. Let $\phi_t(x_0)$ be the solution at $x_0$ of the system \eqref{eq:general} with $\mathbb T=\mathbb R$. For a measurable set $Z$, assume that $\phi_\tau(Z):=\{\phi_\tau(x)\mid x\in Z \}\subset D$ for all $\tau\in[0,t]$. Then
\begin{equation}
 \label{cond:Liouville-like}
 \int_{\Phi_t(Z)}\!\!\rho(x)dx-\int_{Z}\!\!\rho(x)dx=\int_{0}^t\int_{\Phi_\tau(Z)}\!\!\left[\nabla\cdot( F\rho)\right](x)dx d\tau,
\end{equation}   
\end{lemma}
This leads to the following lemma:
\begin{lemma}
\label{lem:contdisc}
Assume that $\rho$ is a continuous-time Lyapunov density for $A$ for the system \eqref{eq:general}, then for any $\tau>0$, $\rho$ is a discrete-time Lyapunov density for $A$ for the time-$\tau$ map $\phi_\tau$ of \eqref{eq:general} 

%Assume that the solutions of \eqref{sys:continuous} exist for all time, uniquely for all initial conditions. Let $\{\phi_t\}_{t\in\mathbb R}$ be the continuous flow generated by \eqref{sys:continuous} and $A$ be a closed invariant set of \eqref{sys:continuous}, namely $\phi_t(A)=A$ for all $t\in\mathbb R$. For each $t>0$, let $\PP_t$ be the Frobenius-Perron operator for the map $\phi_t$ restricted to $A^c$. Let $\rho\in\mathbb C^1(A^c,\mathbb R^n)$ be integrable away from $A$ and satisfy
%$$\nabla\cdot(f\rho)>0.$$ 
% Then for any $t>0$, $$\PP_t\rho(x)<\rho(x)$$ for almost all $x\in A^c$.
\end{lemma}
\begin{proof}
Let $Z\subset A^c$ be a compact set contained in $A^c_\eps$ for some $\eps>0$. 
Choose a $T \in\mathbb R$.

Since $Z$ and $[-T,0]$ are compact sets and the flow map $\phi(t,x) \equiv \phi_t(x) \equiv \phi_x(t)$ is continuous in $t$ and $x$, the set $\phi([-T,0],Z)$ is compact. By the invariance of $A$, any trajectory $\phi_t(z)$, $z\in Z$ is outside $A$, therefore $A$ and $\phi([-T,0],Z)$ are disjoint. Since $\mathbb R^n$ is a normal space, there exists a neighborhood of $A$ that is disjoint from $\phi([-T,0],Z)$. Hence, there exists a $\bar\eps>0$ such that $\phi_\tau(Z)\subset A^c_{\bar\eps}$ for all $\tau\in[-T,0]$.

%By the continuity of \eqref{sys:continuous} with respect to initial states, there \includegraphics{UPAD_Document_20161220_153556}exists a $\bar\eps>0$ such that $\phi_\tau(Z)\subset A^c_{\bar\eps}$ for all $\tau\in[-T,0]$. To see this, assume the contrary, namely that for each $\bar\eps>0$, $\phi_\tau(Z)\cap A_{\bar\eps}\neq \emptyset$ for some $\tau\in[-T,0]$. In other words, there exist sequences $\{x_k\}\in Z$ and $\{t_k\}\in[-T,0]$ such that $\phi_{t_k}x_k\to A$. Since $Z$ and $[-T,0]$ are compact, we can choose subsequences $\{t_{k_m}\}\subset \{t_k\}$ and $\{x_{k_m}\}\subset \{x_k\}$ such that $t_{k_m}\to \bar t\in [-T,0]$ and $x_{k_m}\to \bar x\in Z$. Since the sequence $\phi_{t_{k_m}}x_{k_m}$ converges to the compact set $A$, this sequence has a subsequence that converges to a point $y$ in $A$. However, continuity of \eqref{sys:continuous} with respect to initial states (see Theorem~3.5 in Khalil 3rd ed.), implies that $\phi_{\bar t} \bar x=y$, which contradicts to the fact that $A$ is invariant.

%Choose a subsequence $\{t_{k_m}\}\subset \{t_k\}$ that converges to a number $\bar t\in[-T,0]$, and consider the subsequence $\{x_{}\}$ (See Theorem~3.5 in Khalil 3rd ed.)

Now, we can apply Lemma~\ref{lem:Rantzer} with $t=-T$ and $D=A^c_{\bar\eps}$. Since $\nabla(F\rho)>0$ and $\int_{\phi_{-T}(Z)}\rho(x)dx=\int_{Z}\PP\rho(x)dx$, \eqref{cond:Liouville-like} implies that $\int_Z\PP\rho(x)dx<\int_Z\rho(x)dx$. Since $Z$ is an arbitrary compact subset of $A^c_\eps$, $\PP\rho(x)<\rho(x)$ for almost all $x\in A^c_\eps$. Since $\eps>0$ is also arbitrary, we have $\PP\rho(x)<\rho(x)$ for almost all $x\in A^c$. 
\end{proof}

%\begin{figure}[hpbt]
    %\centering
%\includegraphics[width=.45\textwidth]{UPAD_Document_20161220_153556}
%\caption{Illustration of the proof of Theorem~\ref{th:continuous} }\label{fig:IllustrationThmContinuous}
%\end{figure}

%\begin{figure}[hpbt]
    %\centering
%\includegraphics[width=.45\textwidth]{set}
%\caption{Illustration of the proof of Theorem~\ref{th:continuous} }\label{fig:IllustrationThmContinuous}
%\end{figure}

\begin{figure}[hpb]
    \centering
		\makebox{\parbox{3in}{
		%\def\svgwidth{2in}
		%\fontsize{20.7}{24.9}\selectfont
		\fontsize{15}{0}\selectfont
		\scalebox{.6}{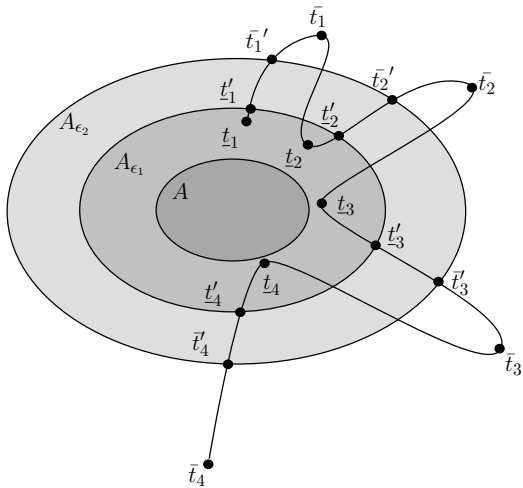}
		%\resizebox{.3\totalheight}{!}{\input{slidingbox.pdf_tex}}
		%\resizebox{.3\linewidth}{!}{\input{slidingbox.pdf_tex}}
    %\input{slidingbox.pdf_tex}
		}}
    \caption{Illustration of the proof of Theorem~\ref{th:continuous} }\label{fig:IllustrationThmContinuous}
\end{figure}

We are now ready for the proof of Theorem~\ref{th:continuous}:
\begin{proof}[Proof of Theorem~\ref{th:continuous}]
	We pick a positive number $\tau$, and define a sequence $(\tau_n = \frac{\tau}{n})$. By Lemma~\ref{lem:contdisc}, $\rho$ is a discrete-time Lyapunov density for the flow maps $\phi_{\tau_n}(x)$ for $n=1,2,\dots$. For each $n$, Theorem~\ref{th:discrete} implies that there exists a set $N_n$ of zero measure such that $x\in (N_n)^\text c$ implies
$\phi_{k\tau_n}(x) \to A$ as $k \to \infty$. Define $N := \bigcup_{n=1}^\infty N_{n}$, which has zero measure. Then, for any $n$ and for any $x\in N^c$, $\phi_{k\tau_n}(x) \to A$ as $k \to \infty$. Hence, for any $\epsilon_1 >0$, there exits $M(n)$ such that for $k > M(n)$, $d(\phi_{k\tau_{n}}(x), A) < \epsilon_1$. 
		
	We will show that $\phi_t(x)\to A$ for all $x\in N^c$. For some $x\in N^c$, suppose that $\phi_t(x)$ does not converge to $A$ as $t \to \infty$. Then, there is $\epsilon_2 >0$ and a sequence  $({t}_k)$  such that $\lim_{k \to \infty} {t}_k = \infty$,  and $\phi_{{t}_k}(x) \in A_{\eps_2}^c$. Pick $\eps_1<\eps_2$, a subsequence $(\underline{t}_k)$ of $(k\tau)$ and a subsequence $(\overline{t}_n)$ of $({t}_n)$ such that $\overline{t}_n \in ]\underline{t}_n,\underline{t}_{n+1}[$, $\phi_x(\underline{t}_n) \in A_{\epsilon_1}$. The situation is illustrated in Fig.\ref{fig:IllustrationThmContinuous}.
	Such subsequences can be chosen as $\lim_{k \to \infty} k \tau = \infty$ and for $k > M(1)$, $\phi_x(k \tau) \in A_{\epsilon_1}$.

	Thanks to the continuity of the flow map, each of the times $\overline{t}_k$ can be decreased by $ \overline{\delta}_k$ such that $\phi_x(\overline{t}_k - \overline{\delta}_k) \in \partial A_{\epsilon_2}$, furthermore each of the times  $\underline{t}_k$ can be increased by $\underline{\delta}_k$ such that $\phi_x(\underline{t}_k + \underline{\delta}_k) \in \partial A_{\epsilon_1}$ and $\phi_x(t) \in A_{\epsilon_2}-A_{\eps_1}$ for all $t\in]\underline{t}_n + \underline{\delta}_n, \overline{t}_{n} - \overline{\delta}_n[$. Let us define $\underline{t}_n' \equiv \underline{t}_n + \underline{\delta}_n$, and $ \overline{t}_{n}' \equiv \overline{t}_{n} - \overline{\delta}_n$.

	Since $A$ is compact, the closure $\overline{A}_{\epsilon_2}$ is compact, and there is a Lipschitz constant of the vector field $F$ on $\overline{A}_{\epsilon_2}$. Hence, $\inf\{\overline{t}_n'-\underline{t}_n'\} > 0$. 
	Pick now $\tau_m\in\{\tau_n\}$ such that $\tau_m < \inf\{\overline{t}_n'-\underline{t}_n'\}$. Then there is an infinite sequence $(k_n) \subset \mathbb N$ with $k_n \tau_m \in ]\underline{t}_n', \overline{t}_{n}'[$ for each $n \in \mathbb N$, such that $\phi_{k_n \tau_m}\in A_{\epsilon_2}-A_{\eps_1}$ for all $n$.  But $\lim_{n \to \infty} d(\phi_{k_n\tau_m}(x), A) = 0$; hence, we arrive at a contradiction.
\end{proof}

%\begin{proposition}
%If $\rho$ is a continuous-time Lyapunov density of an invariant set $A$ for the continuous time system system \eqref{eq:general}, then $\frac{\rho \|F\|}{|x|}$ is integrable away from $A$.
%\end{proposition}

%% file: cp3.pdf_tex
%% Creator: Inkscape 0.91_64bit, www.inkscape.org
%% Accompanies image file 'cp3.pdf' (pdf, eps, ps)
\begingroup%
  \makeatletter%
  \providecommand\color[2][]{%
    \errmessage{(Inkscape) Color is used for the text in Inkscape, but the package 'color.sty' is not loaded}%
    \renewcommand\color[2][]{}%
  }%
  \providecommand\transparent[1]{%
    \errmessage{(Inkscape) Transparency is used (non-zero) for the text in Inkscape, but the package 'transparent.sty' is not loaded}%
    \renewcommand\transparent[1]{}%
  }%
  \providecommand\rotatebox[2]{#2}%
  \ifx\svgwidth\undefined%
    \setlength{\unitlength}{367.0697998bp}%
    \ifx\svgscale\undefined%
      \relax%
    \else%
      \setlength{\unitlength}{\unitlength * \real{\svgscale}}%
    \fi%
  \else%
    \setlength{\unitlength}{\svgwidth}%
  \fi%
  \global\let\svgwidth\undefined%
  \global\let\svgscale\undefined%
  \makeatother%
  \begin{picture}(1,0.80976841)%
    \put(0.38678136,-0.13797608){\color[rgb]{0,0,0}\makebox(0,0)[lb]{\smash{}}}%
    \put(0.25601604,-0.13797608){\color[rgb]{0,0,0}\makebox(0,0)[lb]{\smash{}}}%
    \put(0.3213987,-0.20335874){\color[rgb]{0,0,0}\makebox(0,0)[lb]{\smash{}}}%
    \put(0.12525072,-0.07259348){\color[rgb]{0,0,0}\makebox(0,0)[lb]{\smash{}}}%
    \put(0,0){\includegraphics[width=\unitlength,page=1]{cp3.pdf}}%
    \put(0.18352585,0.54207901){\color[rgb]{0,0,0}\makebox(0,0)[lb]{\smash{$A_{\epsilon_1}$}}}%
    \put(0.08886412,0.6116444){\color[rgb]{0,0,0}\makebox(0,0)[lb]{\smash{$A_{\epsilon_2}$}}}%
    \put(0.28309374,0.49097427){\color[rgb]{0,0,0}\makebox(0,0)[lb]{\smash{$A$}}}%
    \put(0,0){\includegraphics[width=\unitlength,page=2]{cp3.pdf}}%
    \put(0.36882388,0.58827773){\color[rgb]{0,0,0}\makebox(0,0)[lb]{\smash{$\underline{t}_1$}}}%
    \put(0.51847444,0.7911401){\color[rgb]{0,0,0}\makebox(0,0)[lb]{\smash{$\bar{t_1}$}}}%
    \put(0.4094053,0.74916396){\color[rgb]{0,0,0}\makebox(0,0)[lb]{\smash{$\bar{t_1}'$}}}%
    \put(0.36537663,0.66771081){\color[rgb]{0,0,0}\makebox(0,0)[lb]{\smash{$\underline{t}_1'$}}}%
    \put(0.47711617,0.55323611){\color[rgb]{0,0,0}\makebox(0,0)[lb]{\smash{$\underline{t}_2$}}}%
    \put(0.80733159,0.67761727){\color[rgb]{0,0,0}\makebox(0,0)[lb]{\smash{$\bar{t_2}$}}}%
    \put(0.62500187,0.68330224){\color[rgb]{0,0,0}\makebox(0,0)[lb]{\smash{$\bar{t_2}'$}}}%
    \put(0.53914589,0.62606491){\color[rgb]{0,0,0}\makebox(0,0)[lb]{\smash{$\underline{t}_2'$}}}%
    \put(-0.10722097,0.786407){\color[rgb]{0,0,0}\makebox(0,0)[lb]{\smash{}}}%
    \put(0.56652089,0.47288369){\color[rgb]{0,0,0}\makebox(0,0)[lb]{\smash{$\underline{t}_3$}}}%
    \put(0.85311648,0.19643045){\color[rgb]{0,0,0}\makebox(0,0)[lb]{\smash{$\bar{t}_3$}}}%
    \put(0.76163516,0.33749535){\color[rgb]{0,0,0}\makebox(0,0)[lb]{\smash{$\bar{t}_3'$}}}%
    \put(0.65227458,0.41876715){\color[rgb]{0,0,0}\makebox(0,0)[lb]{\smash{$\underline{t}_3'$}}}%
    \put(0.43919916,0.33659539){\color[rgb]{0,0,0}\makebox(0,0)[lb]{\smash{$\underline{t}_4$}}}%
    \put(0.31146879,0.00578244){\color[rgb]{0,0,0}\makebox(0,0)[lb]{\smash{$\bar{t}_4$}}}%
    \put(0.31749532,0.23301784){\color[rgb]{0,0,0}\makebox(0,0)[lb]{\smash{$\bar{t}_4'$}}}%
    \put(0.33912021,0.3197932){\color[rgb]{0,0,0}\makebox(0,0)[lb]{\smash{$\underline{t}_4'$}}}%
    \put(-0.19637914,0.66312656){\color[rgb]{0,0,0}\makebox(0,0)[lb]{\smash{}}}%
    \put(0,0){\includegraphics[width=\unitlength,page=3]{cp3.pdf}}%
  \end{picture}%
\endgroup%

%% file: Examples.tex
\section{Examples}
\label{sec:examples}
In this section, we present examples for Theorem~\ref{th:discrete} and Theorem~\ref{th:continuous}. 
Example~\ref{ex1discrete} and Example~\ref{ex1} are illustrative applications of the theory in one-dimensional case for discrete time and continuous time, respectively. Example~\ref{ex1} also shows that a function $\rho$ may satisfy the integrability condition given in Theorem~\ref{th:continuous} but do not satisfy the integrability condition given in \cite{Rantzer2001}. Finally, Example~\ref{ex2} demonstrates an application of the theory to an invariant set which is a heteroclinic cycle.

%\begin{example}
%Consider the system
%\begin{equation}
%\label{eq:ex1}
%\dot x=-|x|^\alpha \cdot \text{sign}(x),
%\end{equation}
% where $\alpha>1$. Here, we consider the invariant set $A=\{0\}$. Note that the solutions of \eqref{eq:ex1} are well-defined for all time and for all initial conditions, since the right-hand side is locally Lipschitz and solutions are bounded ($\dot x$ and $x$ have different signs).  
% $\rho(x)=\frac{1}{|x|^\beta}$ is a Lyapunov density for \eqref{eq:ex1} if $\beta>\alpha$, since $\nabla\cdot(F\rho)=\frac{d}{dx}\left( \frac{-\text{sign}(x)}{|x|^{\beta-\alpha}} \right)=\frac{1}{|x|^{\beta-\alpha+1}}>0$ outside $A$. Note that both $\rho$ and $\rho \|F\|/|x|$ are integrable. Hence, by Theorem~\ref{th:continuous}, the equilibrium $A$ is almost everywhere stable for positive $\alpha$.
%\end{example}

\begin{example}
\label{ex1discrete}
Consider the discrete-time dynamical system defined on $\mathbb R$
\begin{equation}
\label{eq:ex1discrete}
x(k+1)=T(x(k)),
\end{equation}
 where $T$ is the functional inverse of $S(x)=x+\alpha x^3$ for some $\alpha>0$. Here, we consider the invariant set $A=\{0\}$. It is easy to check that $A$ is globally asymptotically stable (by drawing $T$ and cobwebbing). Let us define $\rho(x)=\frac{1}{|x|^3}$. Then, $\PP \rho(x)=\rho(T^{-1}(x))/(dT(x)/dx)=\rho(S(x))\cdot (dS(x)/dx)=\frac{1+3\alpha x^2}{|x+\alpha x^3|^3}=\rho(x)\cdot \frac{1+3\alpha x^2}{(1+\alpha x^2)^3}$. Here, $\rho$ is a Lyapunov density, since it is positive on $\mathbb R-A$, integrable away from $A$ and it satisfies $\PP \rho (x)<\rho(x)$ on $\mathbb R-A$. The last inequality follows from
$(1+\alpha x^2)^3=1+3\alpha^2 x^4 + 3\alpha x^2 +\alpha^3 x^6>1+3\alpha x^2$.
\end{example}

\begin{example}
\label{ex1}
Consider the system
\begin{equation}
\label{eq:ex1}
\dot x=- \text{sign}(x) \cdot |x|^\alpha (1-e^{-|x|}),
\end{equation}
 where $\alpha>1$. Here, we consider the invariant set $A=\{0\}$. Note that the solutions of \eqref{eq:ex1} are well-defined for all time and for all initial conditions, since the right-hand side is locally Lipschitz and solutions are bounded ($\dot x$ and $x$ have different signs).  
 $\rho(x)={|x|^{-\alpha}}$ is a Lyapunov density for \eqref{eq:ex1}, since $\nabla\cdot(F\rho)=\frac{d}{dx}\left( {-\emph{sign}(x)(1-e^{-|x|})} \right)=e^{-|x|}>0$ outside $A$. Since $\alpha>1$, $\rho$ is integrable away from $A$. Hence, by Theorem~\ref{th:continuous}, the equilibrium $A$ is almost globally stable for positive $\alpha$. Note that, $\rho \|F\|/|x|=\frac{(1-e^{-|x|})}{|x|}$ is not integrable away from $A$, therefore, it does not satisfy the integrability condition in \cite{Rantzer2001}.
\end{example}

\begin{example}[A heteroclinic attractor (\cite{Rodrigues2016})]
\label{ex2}
Consider the following perturbed Hamiltonian system on $\mathbb R^2$ (shown in Fig.~\ref{fig:exam2}):
\begin{subequations}\label{eq:exam2}
\begin{align}
\dot x&=-y\\
\dot y&=x-x^3-\eps y\left( v(x,y)-\frac{1}{4} \right),
\end{align}
\end{subequations}
where $v(x,y)=\frac{2x^2-x^4+2y^2}{4}$ is the Hamiltonian for the system with $\eps=0$.

\begin{figure}
\centering
\makebox{\parbox{3in}{
\subfloat[$\epsilon=0$]{%
\resizebox*{7cm}{!}{\includegraphics{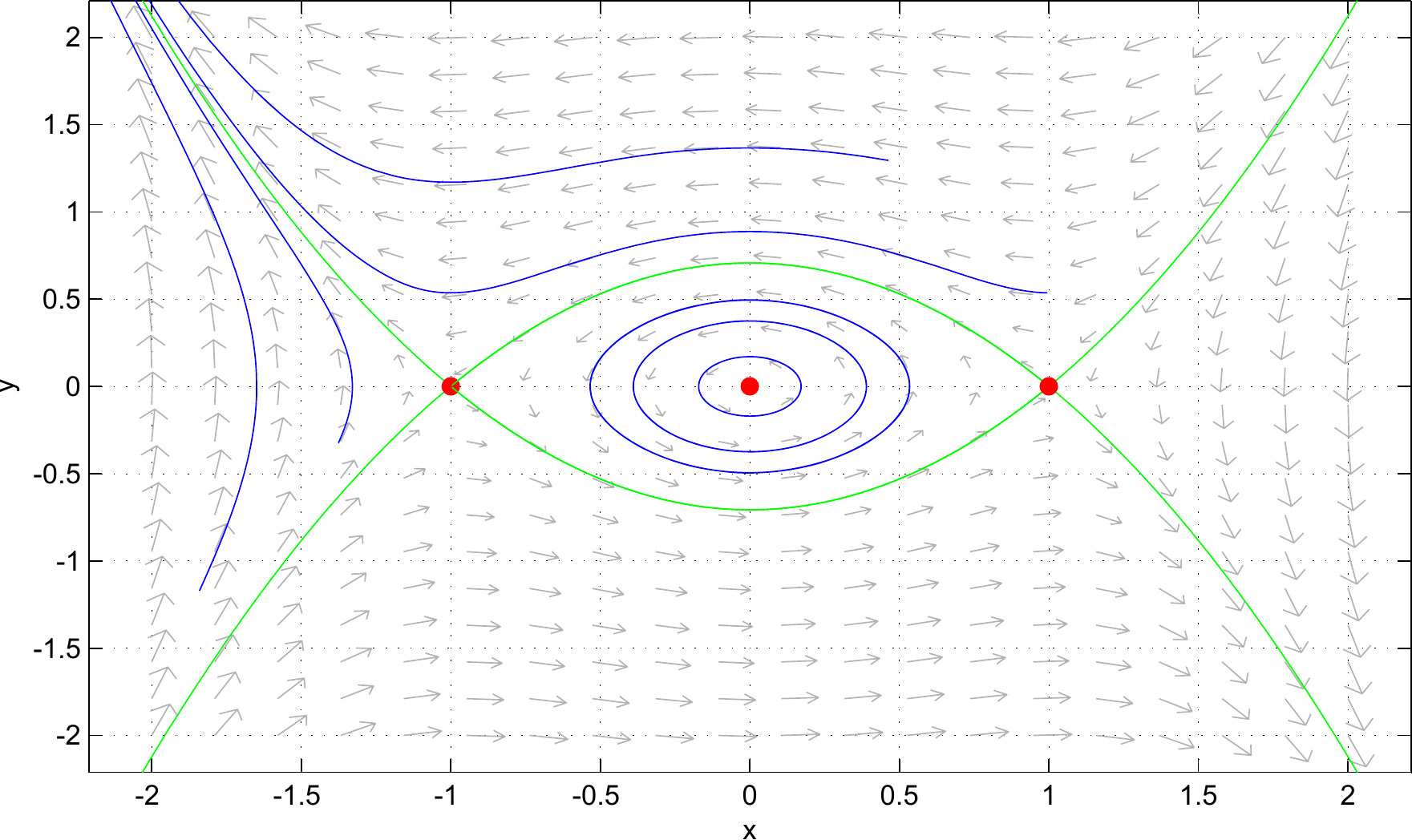}}}\hspace{5pt}
\subfloat[$\epsilon=1$]{%
\resizebox*{7cm}{!}{\includegraphics{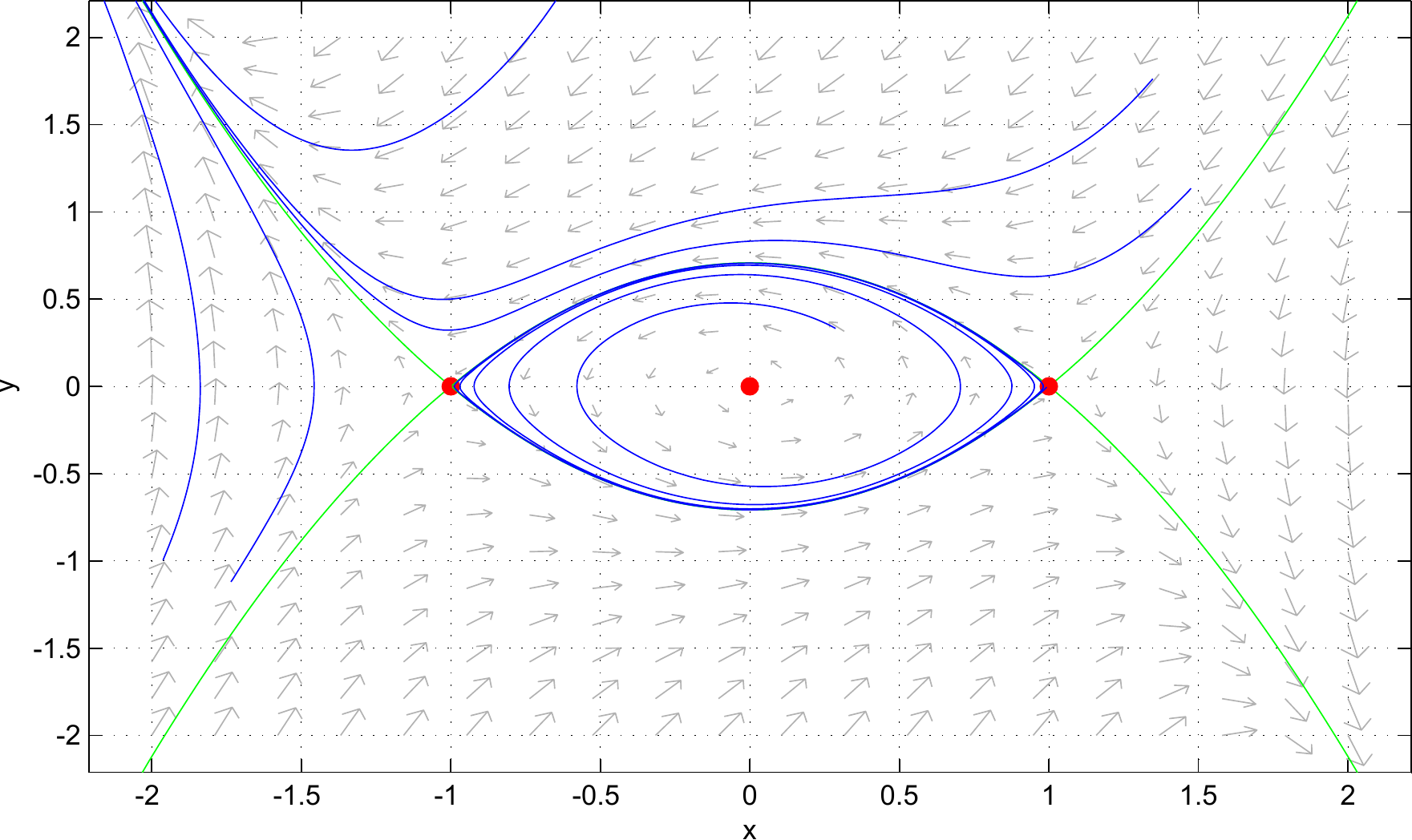}}}
}}
\caption{Heteroclinic trajectories (green), equilibrium points (red), and forward trajectories (blue) for the system given by \eqref{eq:exam2}}\label{fig:exam2}
\end{figure}

%\begin{figure}[hpbt]
    %\centering
%\subfigure[$\epsilon=0$]
%{
    %\includegraphics[width=.45\textwidth]{e0tra1.pdf}\label{fig:sub1}
		%}
%\subfigure[$\epsilon=1$]
%{
    %\includegraphics[width=.45\textwidth]{e1tra1.pdf}\label{fig:sub2}
		%}\label{b}
%\caption{Heteroclinic trajectories (green), equilibrium points (red), and forward trajectories (blue) for the system given by \eqref{eq:exam2}}\label{fig:exam2}
%\end{figure}

For all $\eps$, the system has three equilibria $p=(-1,0)$, $o=(0,0)$ and $q=(1,0)$. The equlibria $q$ and $p$ are contained in the level set $v(x,y)=\frac{1}{4}$. This level set also contains two heteroclinic trajectories connecting $p$ to $q$ and $q$ to $p$, forming a heteroclinic cycle, denoted by $A$, that persists for all $\eps$. We consider the almost global stability of $A$. In \cite{Rodrigues2016}, it is shown that $A$ attracts almost all initial points inside the cycle if $\eps>0$ as stated . We will prove this by considering the Lyapunov density 
\begin{equation}
\rho(x,y)=\frac{1}{\frac{1}{4}-v(x,y)}
\end{equation}
Note that we consider the system on the invariant closed set $\overline W=W\cup A$, where 
%$A$ is the heteroclinic cycle given by the level set $v(x,y)=\frac{1}{4}$ and 
$W$ is the sub-level set $v(x,y)<\frac{1}{4}$. To prove that the heteroclinic cycle $A$ is attracting for almost all initial conditions in $W$, observe that $\rho(x)>0$ for almost all $x\in W$ (LD1), $\rho$ is integrable on $W$ away from $A$ (LD2) and $\nabla\cdot(\rho F)>0$ almost everywhere on $W$ (LD3), since
\begin{eqnarray*}
\nabla\cdot(\rho F)&=&\nabla\cdot\begin{pmatrix} \frac{-y}{\frac{1}{4}-v(x,y)} \\ \frac{x-x^3-\eps y(v(x,y)-\frac{1}{4})}{\frac{1}{4}-v(x,y)} \end{pmatrix}
%=\nabla\cdot \begin{pmatrix} \frac{-4y}{1-2x^2+x^4-2y^2} \\ \frac{4x-4x^3}{1-2x^2+x^4-2y^2}+\eps y \end{pmatrix}
\\
&=&\frac{4y(-4x+4x^3)}{{(1-2x^2+x^4-2y^2)}^2}\\
&&+\frac{(4x-4x^3)4y}{{(1-2x^2+x^4-2y^2)}^2}+\eps=\eps>0.
\end{eqnarray*}

As an ending remark, we point out that this example is essentially a closed loop version of the cart-pole system (or swinging pendulum on a cart). Indeed, for this system the control task is to swing up and balance a pendulum attached to a cart by means of controlling the horizontal movement of the cart. The swing up phase is usually determined by a control strategy which makes the pendulum asymptotically approach to a heteroclinic cycle exactly as depicted in Fig.~\ref{fig:exam2}(b). In this case, Fig.~\ref{fig:exam2}(b) should be interpreted as the phase space $x=\theta$ and $y=\dot{\theta}$, with $\theta$ the angle locating the pendulum and the equilibria $q$ and $p$ having numerical values $(0,0)$ and $(\pi,0)$, respectively. For a detailed account we refer to \cite[Chapter 3]{Fantoni2001}, and references therein.  

\end{example}

%% file: Conclusion.tex
%\section{Remarks on the Different Types of Milnor Attractor}

%Classification of Milnor Attractors

%If f is Milnor and satisfy some property, then it is asymp glob stable.

\section{Conclusion}

{
%NEW - MORE MODEST
We have presented a Lyapunov density theorem for discrete-time systems without assuming compactness of state space and local stability of the invariant set. We have also obtained a new continuous-time Lyapunov density theorem for systems with well-defined solutions on the whole time. 

%OLD
 
%We have revisited  and unified the theory of almost global stability based on Lyapunov density functions, both for the continuous-time and discrete-time cases. In conclusion, we have obtained more general results for almost globally stability of invariant sets and remove the redundant assumptions in literature for the discrete-time case, namely the local stability of the attractor and the compactness of the state space. 
}

\section*{Acknowledgement}

{
Authors thank Alexandre Rodrigues for his explanations about the model in Example~\ref{ex2} and Ferruh \.Ilhan for his useful comments on this paper. 
}